\documentclass[journal]{IEEEtran}

\usepackage{ifpdf}
 \ifpdf
 \else
 \fi

\usepackage{cite}

\ifCLASSINFOpdf
   \usepackage[pdftex]{graphicx}
   \usepackage{epstopdf}
\else
   \usepackage[dvips]{graphicx}
\fi
\usepackage[cmex10]{amsmath}
\usepackage{amssymb}
\usepackage{amsthm}
\begin{document}
%
\title{Performance of General STCs over Spatially Correlated MIMO Single-keyhole Channels}
%
%
%

\author{Chen~He, Xun~Chen,
        and
        Z.~Jane~Wang,~\IEEEmembership{Senior Member,~IEEE}

        \thanks{The authors are with \IEEEauthorblockA{Department of Electrical and Computer Engineering, University of British Columbia, Vancouver, BC, V6T 1Z4 Canada. Emails: \{chenh, xunchen, zjanew\}@ece.ubc.ca.} }
}

%
%

\markboth{}%
{Shell \MakeLowercase{\textit{et al.}}: Bare Demo of IEEEtran.cls for Journals}
%



\maketitle

\begin{abstract}
For MIMO Rayleigh channels, it has been shown that transmitter correlations always degrade the performance of general space-time codes (STCs) in high SNR regimes. In this correspondence, however, we show that when MIMO channels experience single-keyhole conditions, the effect of spatial correlations between transmission antennas is more sophisticated for general STCs: when $M>N$ (i.e., the number of transmission antennas is greater than the number of receiving antennas), depending on how the correlation matrix $\mathbf{P}$ beamforms the code word difference matrix $\mathbf{\Delta}$, the PEP performance of general STCs can be either degraded or improved in high SNR regimes. We provide a new measure, which is based on the eigenvalues of $\mathbf{\Delta}$ and the numbers of transmission and receiving antennas, to exam if there exists certain correlation matrices that can improve the performance of general STCs in high SNR regimes. Previous studies on the effect of spatial correlations over single-keyhole channels only concentrated on orthogonal STCs, while our study here is for general STCs and can also be used to explain previous findings for orthogonal STCs.
\end{abstract}


\begin{IEEEkeywords}
Space-time codes, spatial correlations, MIMO, keyhole fading.
\end{IEEEkeywords}

%
\IEEEpeerreviewmaketitle

\newtheorem{Theorem}{Theorem}
\newtheorem{Lemma}{Lemma}
\newtheorem{Proposition}{Proposition}

\section{Introduction}\label{Sec: Introduction}
It is well known that multiple-input and multiple-output (MIMO) systems can
achieve better spectral efficiency, without requiring extra power consumption and bandwidth expansion.
In reality, however, individual antennas could be spatially correlated due to insufficient
antenna spacing\cite{Chuah2002}. The effect of spatial correlation on general space-time codes (STCs) has been extensively studied for Rayleigh fading channels, and it has been shown that in asymptotically high signal-to-noise ratio (SNR) regimes, the transmission correlations always degrade the pair-wise error probability (PEP) performance \cite{Wang2004} \cite{Bolcskei2000}, while only in the asymptotically low SNR regimes, the transmission correlations may either improve or degrade the PEP performance \cite{Wang2004}.

In MIMO channels, if the scattering environment is not-so-rich, it is demonstrated in \cite{Chizhik2002} \cite{Gesbert2002} that MIMO channels can experience single-keyhole conditions,
where despite rich local scattering and independent transmitting and receiving signals, the system only has a rank-1 MIMO channel matrix. A number of papers have studied the performance of single-keyhole channels
\cite{Shin2003A, Levin2005, Muller2007, Gong2007, Sanayei2007, Yahampath2007} and its extension: multi-keyhole conditions \cite{Shin2008, Zhong2011, Yang2011, Levin2008, Levin2011}.

In the literature, the effect of spatial correlation on the performance of STCs over the single-keyhole channel has only been investigated for orthogonal space-time block codes (OSTBCs). It was shown that the for OSTBCs, the spatial correlations between transmission antennas always degrade the PEP performance \cite{Yahampath2007} \cite{Shin2008}.
Particularly, \cite{Yahampath2007} investigated OSTBCs with linear precoding, and showed that the correlations between transmission antennas can only degrade the performance in single-keyhole conditions.
While in \cite{Shin2008}, researchers used
majorization relations of the correlation matrices to show
that for OSTBCs, the correlations always degrade the PEP performance in multi-keyhole conditions, and this result is also applicable to the single-keyhole case.
However, when the STCs are not orthogonal, the effect of the spatial correlations on the PEP performance is still not clear in single-keyhole channels.
In this correspondence, we address this issue.  Instead of using majorization relations, which cannot be applied to non-orthogonal STCs in single-keyhole channels, we provide a new measure to exam if there exists a certain correlation matrix that can improve the PEP performance of general STCs in single-keyhole channels. We will prove that, very different from orthogonal codes, when the number of transmission antennas is greater than the number of receiving antennas, the PEP performance of non-orthogonal STCs \emph{can be improved by the transmission correlations in single-keyhole conditions,} \emph{even in the high SNR regimes}. This depends on how the correlation matrix $\mathbf{P}$ beamforms the code difference matrix $\mathbf{\Delta}$.
The major results of this correspondence can be summarized as follows:
\begin{itemize}
  \item In the high SNR regimes, when $M>N$ (the number of transmission antennas is greater than the number of receiving antennas), depending on how the correlation matrix beamforms the code word difference matrix, the correlations can either degrade or improve the PEP performance. Particularly we provide a new measure: let $0 \le \lambda_{1} \le \lambda_{2} \le \cdots \le \lambda_{M} $ be the eigenvalues of $\mathbf{\Delta}^H \mathbf{\Delta}$ in ascending order, and $\bar{\lambda}$ be their average, if there is an integer $L$ such that
        \begin{align}\label{Eq: Eig_Condition}
            \bar{\lambda}^N\left(\sum_{i=L+1}^M\frac{\lambda_{i}^{-1}}{M}\right)^N
            <\frac{\Gamma(M-N)\Gamma(M-L)}{\Gamma(M)\Gamma(M-L-N)},
        \end{align}
      we can always find certain correlation matrices that can improve the PEP performance. We also provide one form of such matrices.
\end{itemize}
\begin{itemize}
  \item In the high SNR regimes, when $M\le N$ (the number of transmission antennas is smaller or equal to the number of receiving antennas), the transmission correlations always degrade the PEP performance.
\end{itemize}
It is worth mentioning that OSTBCs can never satisfy the condition in \eqref{Eq: Eig_Condition}, therefore the transmission correlations always degrade the PEP performance, which is consistent with the findings in \cite{Yahampath2007} and \cite{Shin2008}.
The above results for MIMO single-keyhole channels are proved in Theorem \ref{Th: M_greater_N} and
Theorem \ref{Th: M_smaller_N}.

\emph{Notations}: In this correspondence, $\exp(\cdot)$, $\Gamma(\cdot)$, and $Q(\cdot)$ mean the exponential function, the Gamma function, and the Gaussian $Q$ function, respectively; $\mathbb{E}_X(\cdot)$, $X|Y$, $\|\cdot\|_F$, $\|\cdot\|$, $|\cdot|$, $(\cdot)^T$, $(\cdot)^H$, $det(\cdot)$, $R(\cdot)$, and $trace(\cdot)$ denote the expectation over the density of $X$, the conditional random variable of $X$ given $Y$, the Frobenius norm of a matrix, the norm of a matrix, the magnitude of a complex number, the transpose, the conjugate transpose, the determinant, the rank, and the trace of a matrix, respectively; $A \doteq B$ means that $A$ is equal to $B$ in the limit, and $X \sim Y$ means that $X$ is identically distributed with $Y$.

\section{Channel Model}\label{Sec: Channel_Model}
In the general MIMO channel, the signal model is given by
\begin{align} \label{Eq: Signaling_Model}
\mathbf{R}=\sqrt{\frac{\bar{\gamma}}{M}}\mathbf{C}\mathbf{H}+\mathbf{W},
\end{align}
where the $T \times N$ matrix $\mathbf{R}$ represents the received signal,
$\mathbf{C}$ is the $T\times M$ transmitted codeword matrix, $\bar{\gamma}$ is the average
signal-to-noise ratio (SNR), and $\mathbf{W}$ is the zero-mean additive
circularly symmetric complex Gaussian noise matrix with size
$T \times N$,  whose elements have unit variance per dimension.
It is assumed that the channel state information (CSI) is perfectly
known at the receiver and unknown at the transmitter.

When the communication channel experiences the single-keyhole condition, the channel matrix $\mathbf{H}$ can be modeled by
\begin{align} \label{Eq: Single_Keyhole_Channel_Matrix}
\mathbf{H}=\mathbf{P}^{\frac{1}{2}}\mathbf{h}\mathbf{g}^T.
\end{align}
where the elements of the vector $\mathbf{h}$, i.e. $h_{m}$'s ($n=1,...,M$) represent the normalized channel gains from $M$ transmitting antennas to the keyhole, and the elements of the vector $\mathbf{g}$, i.e. $g_{n}$'s ($n=1,...,N$)  represent the normalized channels from the keyhole to $N$ receiving antennas, and both $h_{m}$'s and $g_{n}$'s can be modeled as complex Gaussian random variable with zero mean and unity variance.
The $M \times M$ matrix $\mathbf{P}$ is the correlation matrix of the transmission antennas. When $\mathbf{P}=\mathbf{I}$ (the identity matrix), the model given in \eqref{Eq: Single_Keyhole_Channel_Matrix} reduces to the case that transmitting antennas are spatially independent, as studied in \cite{Sanayei2007}.

\section{Effect of Transmission Correlations on Space-time Code Performance}
PEP, the probability of transmitting codeword $\mathbf{C}$ over $T$ time slots and deciding in favor of another codeword $\mathbf{\mathbf{C}'}$  at the decoder, generally serves as a design criteria for STCs.
When signals transmit over a fading channel with channel matrix $\mathbf{H}$, the code words distance between $\mathbf{C}$ and $\mathbf{C}'$ is defined by the random variable $\|\mathbf{\Delta}\mathbf{H}\|_F$, where $\mathbf{\Delta} \triangleq \mathbf{C}-\mathbf{C}'$ is the code word difference matrix, and $\|\cdot\|_F$ is the Frobenius norm. The PEP of a Gaussian noise channel can be evaluated by averaging the density of $\|\mathbf{\Delta}\mathbf{H}\|_F$ over the $Q$ function as
\begin{align} \label{Eq: PEP_Q}
P(\mathbf{C}\rightarrow \mathbf{C}'|\mathbf{H}) =Q\left(\sqrt{\frac{\bar{\gamma}}{M}\|\mathbf{\Delta}\mathbf{H}\|_F^2}\right).
\end{align}
Using an alternative representation of the $Q$ function, we have
\begin{align} \label{Eq: PEP_exp}
P(\mathbf{C}\rightarrow \mathbf{C}'|\mathbf{H})
=\frac{1}{\pi}\int_{\theta=0}^{\infty}\exp\left(-\frac{\bar{\gamma}}{M}\frac{\|\mathbf{\Delta}\mathbf{H}\|_F^2}{2\sin^2\theta}\right)d\theta.
\end{align}
To find the PEP, we reconsider the code words distance
\begin{align}
\|\mathbf{\Delta}\mathbf{H}\|_F=\|\mathbf{\Delta}\mathbf{P}^{\frac{1}{2}}\mathbf{h}\mathbf{g}^T\|_F,
\end{align}
it is clear that when spatial correlations present between the transmission antennas, the code word difference matrix $\mathbf{\Delta}$ is modified by the correlation matrix $\mathbf{P}$ and becomes
\begin{align}
\mathbf{\Delta}'=\mathbf{\Delta}\mathbf{P}^{\frac{1}{2}}.
\end{align}
Therefore the single-keyhole channel with transmission correlation $\mathbf{P}$ and code word difference matrix $\mathbf{\Delta}$ can be viewed as the single-keyhole channel with spatially independent transmission antennas and code word difference matrix $\mathbf{\Delta}'$. With the above observation, the asymptotic form of PEP for single-keyhole channel with correlation matrix $\mathbf{P}$ can be obtained by using the result for independent transmission antennas over the single-keyhole channel given by \cite{Sanayei2007}: we only need to replace the eigenvalues of $\mathbf{\Delta}$ by $\mathbf{\Delta}'$ for the PEP expression. For distinct $\rho_i$'s, the asymptotic PEP is obtained as
\begin{align} \label{Eq: PEP}
  & P(\mathbf{C}\rightarrow \mathbf{C}') \doteq \nonumber \\
  &  \left\{
  \begin{array}{l l l}
    &C_1\left(\prod_{i=1}^{R(\mathbf{\Delta}')}\rho_{i}\right)^{-1}\bar{\gamma}^{-R(\mathbf{\Delta}')},
    \quad \text{if $N>R(\mathbf{\Delta}')$}; \\ 
    &C_1\left(\prod_{i=1}^{R(\mathbf{\Delta}')}\rho_{i}\right)^{-1} (\ln\bar{\gamma})\bar{\gamma}^{-R(\mathbf{\Delta}')},
    \quad \text{if $N=R(\mathbf{\Delta}')$}; \\ 
    &C_3\sum_{i=1}^{R(\mathbf{\Delta})'}\frac{\ln \rho_{i}}{\rho_{i}^N}\prod_{j\neq i}\frac{\rho_i}{\rho_i-\rho_j}\bar{\gamma}^{-N},
    \quad \text{if $N<R(\mathbf{\Delta}')$}. 
  \end{array} \right.
\end{align}
For identical $\rho_i$'s,  we have
\begin{align} \label{Eq: PEP_identical}
   P(\mathbf{C}\rightarrow \mathbf{C}') \doteq
    \left\{
  \begin{array}{l l l}
    &C_1\rho^{-R(\mathbf{\Delta}')}\bar{\gamma}^{-R(\mathbf{\Delta}')},
    \quad \text{if $N>R(\mathbf{\Delta}')$}; \\ 
    &C_1\rho^{-R(\mathbf{\Delta}')} (\ln\bar{\gamma})\bar{\gamma}^{-R(\mathbf{\Delta}')},
    \quad \text{if $N=R(\mathbf{\Delta}')$}; \\ 
    &C_2 \rho^{-R(\mathbf{\Delta}')}\bar{\gamma}^{-N},
    \quad \text{if $N<R(\mathbf{\Delta}')$}. 
  \end{array} \right.
\end{align}
where $\rho_i$'s are the eigenvalues of $\mathbf{P}^{\frac{H}{2}}\mathbf{\Delta}^H\mathbf{\Delta}\mathbf{P}^{\frac{1}{2}}$, and
\begin{align*}
C_1=
\frac{\Gamma(\frac{1}{2}+R(\mathbf{\Delta}'))}{2\sqrt{\pi}\Gamma(1+R(\mathbf{\Delta}'))}
\times
\frac{N^{R(\mathbf{\Delta}')}\Gamma(N-R(\mathbf{\Delta}'))}{\Gamma(N)},
\end{align*}
\begin{align*}
C_2=
\frac{\Gamma(\frac{1}{2}+N)}{2\sqrt{\pi}\Gamma(1+N)}
\times
\frac{N^{N}\Gamma(R(\mathbf{\Delta}')-N)}{\Gamma(R(\mathbf{\Delta}'))},
\end{align*}
and
\begin{align*}
C_3
=\frac{\Gamma(\frac{1}{2}+N)}{2\sqrt{\pi}\Gamma(1+N)}
\times
(-1)^{N-1}
\times
\frac{1}{\Gamma(N)}.
\end{align*}
Although the asymptotic PEP for correlated transmission antennas has been obtained, our main concern, how the transmission correlation $\mathbf{P}$ affects the performance of space-time codes, is still not clearly answered, especially for the case that $M>N$. Actually, from the expressions in \eqref{Eq: PEP} and \eqref{Eq: PEP_identical}, there is no clue about this. The major work of this correspondence is to address this issue.

To investigate the effect of transmission correlations on the PEP performance, we first present the following facts and inferences about the correlation matrix $\mathbf{P}$ and the code difference matrix $\mathbf{\Delta}$, let $\lambda_{i}$'s denote the eigenvalues of $\mathbf{\Delta}^H\mathbf{\Delta}$. Let $\nu_i$'s denote the eigenvalues of $\mathbf{P}$ and $\rho_{i}$'s denote the eigenvalues of $\mathbf{P}^{\frac{H}{2}}\mathbf{\Delta}^H\mathbf{\Delta}\mathbf{P}^{\frac{1}{2}}$, then we have
\begin{enumerate}
  \item $trace(\mathbf{P})=M$, or equivalently,
  \begin{align}
  \sum_{i=1}^M\nu_i=M,
  \end{align}
  where $\nu_i$'s are the eigenvalues of $\mathbf{P}$. This is because the total transmission power is fixed.
  \item \begin{align}
        \left(\prod_{i=1}^M\nu_i\right)\left(\prod_{i=1}^M\lambda_i\right)=\prod_{i=1}^M\rho_i.
        \end{align}
        This is from the fact that $det(\mathbf{P}^{\frac{H}{2}}\mathbf{\Delta}^H\mathbf{\Delta}\mathbf{P}^{\frac{1}{2}})=det(\mathbf{P}) \times det(\mathbf{\Delta}^H\mathbf{\Delta})$
\end{enumerate}
In this correspondence, it is assumed that the code construction achieves full rank, i.e. $R(\mathbf{\Delta)}=M$. We now start to analyze the effect of correlations on the PEP performance. We consider the cases that $M \le N$ and $M > N$ separately.

\subsection{More Transmission Antennas than Receiving Antennas: $M > N$}
In this case, we first provide the following Lemma:
\begin{Lemma}\label{Lemma: Bound}
Let $\lambda_i$ be real for $i\in \{1,2,...,M\}$ and $\bar{\lambda}=\frac{\sum_{i=1}^M\lambda_i}{M}$. Let $X_i$, $i\in \{1,2,...,M\}$ be a set of i.i.d random variables, $Y$ be another random variable which is independent with $X_i$'s, then we have
\begin{align}\label{Eq: Lemma_Balanced_Eigen_Value}
\mathbb{E}\left(f\left(Y\sum_{i=1}^M \bar{\lambda} X_i\right)\right) \le \mathbb{E}\left(f\left(Y\sum_{i=1}^M \lambda_i X_i\right)\right),
\end{align}
where $f(\cdot)$ is a convex function. The equal sign holds when $\lambda_{i}=\bar{\lambda}$ for all $i\in \{1,2,...,M\}$.
\end{Lemma}

\begin{proof}
To prove \eqref{Eq: Lemma_Balanced_Eigen_Value}, we first prove that \eqref{Eq: Lemma_Balanced_Eigen_Value} holds for any fixed value of $Y$, i.e.
\begin{align} \label{Eq: Lemma_Balanced_Eigen_Value_Fixed_Y}
\mathbb{E}\left(f\left(y\sum_{i=1}^M \bar{\lambda} X_i\right)\right) \le \mathbb{E}\left(f\left(y\sum_{i=1}^M \lambda_i X_i\right)\right),
\end{align}
where $y$ is any possible value that the random variable $Y$ can take. It is easy to see that
\eqref{Eq: Lemma_Balanced_Eigen_Value_Fixed_Y} implies
\eqref{Eq: Lemma_Balanced_Eigen_Value}.

We define
\begin{align}
X\triangleq y\sum_{i=1}^M\bar{\lambda}X_i,
\end{align}
\begin{align}
W\triangleq y\sum_{i=1}^M\lambda_iX_i,
\end{align}
and
\begin{align}
Z\triangleq X-W.
\end{align}
Base on the form of $X$, it is easy to see that the conditional random variables $X_i|X$, $i\in \{1,...,M\}$, are identically distributed, which implies
\begin{align}
\mathbb{E}(X_1|X)=\mathbb{E}(X_2|X)= \cdots = \mathbb{E}(X_M|X).
\end{align}
Therefore
\begin{align}
\mathbb{E}(Z|X)&=\sum_{i=1}^M \bar{\lambda}\mathbb{E}(X_i|X)-\sum_{i=1}^M \lambda_{i}\mathbb{E}(X_i|X)\nonumber\\
      &=\mathbb{E}(X_1|X)\left(\sum_{i=1}^M \bar{\lambda}-\sum_{i=1}^M \lambda_{i}\right)=0.
\end{align}
Since $f(\cdot)$ is convex, by Jensen's inequality we have
\begin{align}
\mathbb{E}(f(X-Z)|X)
&\ge f(\mathbb{E}((X-Z)|X)) \nonumber\\
&=f(X-0)=f(X)
\end{align}
Therefore
\begin{align}
\mathbb{E}(f(W))=\mathbb{E}(\mathbb{E}(f(X-Z)|X))\ge \mathbb{E}(f(X)),
\end{align}
and consequently \eqref{Eq: Lemma_Balanced_Eigen_Value} holds.
\end{proof}


Now we present the main result for the effect of correlations on the PEP performance when $M>N$:
\begin{Theorem}\label{Th: M_greater_N}
In the MIMO single-keyhole channels, when $M>N$, if we can find some integer $L$ between $1$ and $M-N-1$,
i.e. $1 \le L \le M-N-1$ such that
\begin{align}\label{Eq: EigenValueCriteriaCorrelationAffect}
\bar{\lambda}^N\left(\sum_{i=L+1}^M\frac{\lambda_{i}^{-1}}{M}\right)^N
<\frac{\Gamma(M-N)\Gamma(M-L)}{\Gamma(M)\Gamma(M-L-N)},
\end{align}
then there always exist certain correlation matrices that can improve the PEP performance in the asymptotic high SNR regimes.
Here $0 \le \lambda_{1} \le \lambda_{2} \le \cdots \le \lambda_{M} $ are the $eigenvalues$ of $\mathbf{\Delta}^H \mathbf{\Delta}$ in ascending order, and $\bar{\lambda}$ is their average.
\end{Theorem}

\begin{proof}
When the transmission antennas are independent, we have
\begin{align}
\|\mathbf{\Delta} \mathbf{H}\|_F^2=\|\mathbf{g}\|^2\sum_{i=1}^M\lambda_{i}|h_i|^2.
\end{align}
Since the exponential function is convex for real numbers, by applying Lemma \ref{Lemma: Bound} we have
\begin{align}
\mathbb{E}\left(\exp\left(-\|\mathbf{g}\|^2\sum_{i=1}^M\lambda_{i}|h_i|^2\right)\right)
\ge
\mathbb{E}\left(\exp\left(-\|\mathbf{g}\|^2\sum_{i=1}^M\bar{\lambda}|h_i|^2\right)\right),
\end{align}
which implies that when the transmission antennas are independent, the PEP can be bounded as following:
\begin{align}
P_{\mathbf{I}}(\mathbf{C} \rightarrow \mathbf{C}')
\ge \frac{\Gamma(\frac{1}{2}+N)}{2\sqrt{\pi}\Gamma(1+N)} \times \frac{N^N\Gamma(M-N)}{\Gamma(M)}\bar{\lambda}^{-N}\bar{\gamma}^{-N}.
\end{align}
Now suppose that the eigendecompostion of $\mathbf{\Delta}^H \mathbf{\Delta}$ is $\mathbf{UVU}^H$, we consider the following class of correlation matrices for which $\mathbf{P}^{\frac{1}{2}}$ has singular value decomposition as
\begin{align}\label{Eq: CorrelationMatrixImprovePEP}
\mathbf{P}^{\frac{1}{2}}=\mathbf{U}\mathbf{S^{\frac{1}{2}}}\mathbf{D}^H,
\end{align}
where $\mathbf{D}$ is a unitary matrix and the diagonal matrix $\mathbf{S}$ with $\mathbf{S}_{i,i}=\nu_i$ satisfies the power constraint: $\sum_{i=1}^M\nu_i=M$. It follows that
\begin{align}
\|\mathbf{\Delta} \mathbf{P}^{\frac{1}{2}}\mathbf{h} \mathbf{g}^T \|_F^2
&=\|\mathbf{V}^{\frac{1}{2}}\mathbf{U}^H \mathbf{US^{\frac{1}{2}}D}^H \mathbf{h} \mathbf{g}^T \|_F^2\nonumber\\
&=\|\mathbf{V}^{\frac{1}{2}}\mathbf{S}^{\frac{1}{2}}\mathbf{D}^H \mathbf{h} \mathbf{g}^T\|_F^2\nonumber\\
&\sim \|\mathbf{V}^{\frac{1}{2}}\mathbf{S}^{\frac{1}{2}} \mathbf{h} \mathbf{g}^T\|_F^2\nonumber\\
&=\sum_{i=1}^{R({\mathbf{P}})}\rho_{i} \sum_{n=1}^N\|g_n\|^2\|h_i\|^2,
\end{align}
where
\begin{align}
\rho_{i}=\nu_{i}\lambda_{i}
\end{align}
for all $\rho_i$'s.
Now we can have a correlation matrix $\mathbf{P}$ such that
\begin{align}
\nu_1=\nu_2=\cdots=\nu_L=0,
\end{align}
and
\begin{align}
\nu_{L+1},\cdots, \nu_{M}>0.
\end{align}
If we pose another constraint on $\nu_{i}$ such that the non-zero eigenvalues of
$\mathbf{P}^{\frac{H}{2}}\mathbf{\Delta}^H \mathbf{\Delta} \mathbf{P}^{\frac{1}{2}}$ are identical:
\begin{align}
\rho_{L+1}= \rho_{L+2} =\cdots =\rho_{M},
\end{align}
we have
\begin{align}
\rho_{L+1}=\rho_{L+2}=\cdots =\rho_{M}=M\left(\sum_{i=L+1}^{M}\frac{1}{\lambda_{i}}\right)^{-1}.
\end{align}
 In high SNR regimes, the PEP respective to $\mathbf{P}^{\frac{H}{2}}\mathbf{\Delta}^H \mathbf{\Delta} \mathbf{P}^{\frac{1}{2}}$ becomes
\begin{align}
&P_{\mathbf{P}}(\mathbf{C}\rightarrow \mathbf{C}')=\frac{\Gamma(\frac{1}{2}+N)}{2\sqrt{\pi}\Gamma(1+N)} \nonumber\\
& \times \frac{N^N\Gamma(M-L-N)}{\Gamma(M-L)}M^{-N}\left(\sum_{i=L+1}^{M}\frac{1}{\lambda_{i}}\right)^{N}\bar{\gamma}^{-N}.
\end{align}
Therefore $P_{\mathbf{P}}(\mathbf{C}\rightarrow \mathbf{C}')<P_{\mathbf{I}}(\mathbf{C}\rightarrow \mathbf{C}')$ if \eqref{Eq: EigenValueCriteriaCorrelationAffect} holds, i.e. the correlation matrix $\mathbf{P}$ defined in \eqref{Eq: CorrelationMatrixImprovePEP} improves the PEP performance.
\end{proof}

\subsection{Same or Less Transmission Antennas than Receiving Antennas: $M\le N$}
Now we consider the case that the number of transmission antennas is the same as or less than the number of receiving antennas, i.e. $M\le N$.
\begin{Theorem}\label{Th: M_smaller_N}
In the MIMO single-keyhole channels, when $M\le N$, the spatial correlations between transmission antennas always degrade the PEP performance in the high SNR regimes.
\end{Theorem}
\begin{proof}
If $\mathbf{P}$ is rank deficient, because we assume that $\mathbf{\Delta}^H\mathbf{\Delta}$ is full rank, we have
\begin{align}
R(\mathbf{P}^{\frac{H}{2}}\mathbf{\Delta}^H\mathbf{\Delta}\mathbf{P}^{\frac{1}{2}})<R(\mathbf{\Delta}^H \mathbf{\Delta}),
\end{align}
from the PEP given in Equation \eqref{Eq: PEP}, we can see that this will result in a reduction of the diversity order, hence the PEP performance is degraded.
If $\mathbf{P}$ is of full rank, it means $\prod_{i=1}^M\nu_i \neq 0$, and by applying the AM-GM inequality, we have
\begin{align}
\prod_{i=1}^M \nu_{i} \le \left(\frac{\sum_{i=1}^M \nu_i}{M}\right)^M=1,
\end{align}
therefore
\begin{align}
\prod_{i=1}^M\rho_i \le \prod_{i=1}^M\lambda_i.
\end{align}
Note that the equality only holds when $\mathbf{P}$ is an identity matrix. Therefore the PEP is always degraded by transmission correlations for the case that $M\le N$.
\end{proof}


\subsection{Examples and Simulations}
In this Section, we provide an example and perform Monte Carlo simulations for Theorem \ref{Th: M_greater_N} and Theorem \ref{Th: M_smaller_N}.\\
\emph{\textbf{Example}}
We consider some pair of codewords for which
\begin{align}
\mathbf{\Delta}^H\mathbf{\Delta} =
\left(
  \begin{array}{ccc}
   2                & -.95 + .029i     &  -.95 - .029i  \\
   -.95 - .029i   &  2               &  -.95 + .029i \\
   -.95 + .029i   & -.95 - .029i   &  2
  \end{array}
\right),
\end{align}
the eigenvalues are $\lambda_{1}=0.1$, $\lambda_{2}=2.9$ and $\lambda_{3}=3$. Suppose there are three transmission antennas ($M=3$) and one receiving antenna ($N=1$), then there exist an $L=1$ such that
\begin{align}
\bar{\lambda}^N\left(\sum_{i=L+1}^M\frac{\lambda_{i}^{-1}}{M}\right)^N=2\left(\frac{1/2.9+1/3}{3}\right)=0.45,
\end{align}
and
\begin{align}
\frac{\Gamma(M-N)\Gamma(M-L)}{\Gamma(M)\Gamma(M-L-N)}=\frac{\Gamma(3-1)\Gamma(3-1)}{\Gamma(3)\Gamma(3-1-1)}=0.5.
\end{align}
By Theorem \ref{Th: M_greater_N}, there exist some correlation matrices that can improve the PEP performance. One of such matrices can be given by
\begin{align}
\mathbf{P}_1&=\nu_1\mathbf{u}_1\mathbf{u_1}^H+\nu_2\mathbf{u}_2\mathbf{u}_2^H+\nu_3\mathbf{u}_3\mathbf{u}_3^H \nonumber\\
          &=\left(
  \begin{array}{ccc}
   1             &  -.5 - .0144i   &  -.5 + .0144i \\
   -.5 + .0144i  &  1              & -.5 - .0144i \\
   -.5 - .0144i  & -.5 + .0144i    &  1
  \end{array}
\right).
\end{align}
where
\begin{align}
\nu_1=0,
\end{align}
\begin{align}
\nu_2=\frac{M\left(\sum_{i=L+1}^{M}\frac{1}{\lambda_{i}}\right)^{-1}}{\lambda_2}=1.525,
\end{align}
\begin{align}
\nu_3=\frac{M\left(\sum_{i=L+1}^{M}\frac{1}{\lambda_{i}}\right)^{-1}}{\lambda_3}=1.475,
\end{align}
and $\mathbf{u}_1$, $\mathbf{u}_2$, $\mathbf{u}_3$ are the eigenvectors of $\mathbf{\Delta}^H\mathbf{\Delta}$.
From the simulations, we can see that the transmission correlations defined by $\mathbf{P}_1$ can bring about $1.5$ dB gains for the PEP performance, which is illustrated by the square line in Fig. \ref{Fig: PEP_Correlation_Enhance_M3N1}.
 Now we consider another correlation matrix $\mathbf{P}_2$ that has the same eigenvectors as that of $\mathbf{P}_1$ but different eigenvalues: $\nu_1=1.8$, $\nu_2=0.7$ and $\nu_3=0.5$:
\begin{align}\label{Eq: CorrelationMatrix_Degrade_PEP_M3N1}
\mathbf{P}_2=\left(
  \begin{array}{ccc}
   1             & .4 + .0577i   &.4 - .0577i \\
   .4 - .0577i   & 1             &.4 + .0577i  \\
   .4 + .0577i   & .4 - .0577i   & 1
  \end{array}
\right).
\end{align}
For this correlation matrix $\mathbf{P}_2$, which does not satisfy the new measure given in \eqref{Eq: EigenValueCriteriaCorrelationAffect},  we can see that the transmission correlations degrade the PEP performance, as illustrated by the PEP curve (marked by circle) in
Fig. \ref{Fig: PEP_Correlation_Enhance_M3N1}.

Now we keep everything unchanged except that there are four receiving antennas ($N=4$).  By Theorem \ref{Th: M_smaller_N}, all the correlation matrices are expected to degrade the PEP performance in the high SNR regimes, and this is confirmed in Fig. \ref{Fig: PEP_Correlation_Degradtion_M3N4}: both $\mathbf{P}_1$ and $\mathbf{P}_2$ degrade the PEP performance .

\begin{figure}
\centering
  \includegraphics[scale=0.6]{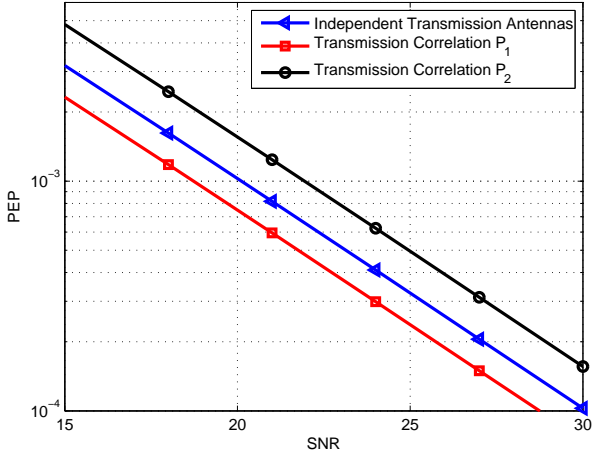}\\
  \caption{Effect of transmission correlations on MIMO single-keyhole channel for the case that $M > N$ (for Theorem \ref{Th: M_greater_N}): in the asymptotically high SNR regime, transmission correlations may either improve or degrade the PEP performance. In this simulation, the correlation matrix $\mathbf{P}_1$ (satisfies the new measure in \eqref{Eq: EigenValueCriteriaCorrelationAffect}) improves the PEP performance, while the correlation matrix $\mathbf{P}_2$ (does not satisfy the new measure in \eqref{Eq: EigenValueCriteriaCorrelationAffect}) degrades the PEP performance.}\label{Fig: PEP_Correlation_Enhance_M3N1}
\end{figure}

\begin{figure}
\centering
  \includegraphics[scale=0.6]{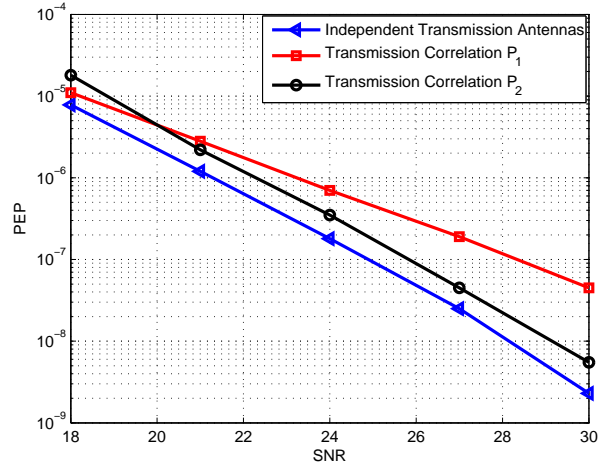}\\
  \caption{Effect of transmission correlations on MIMO single-keyhole channel for the case that $M \le N$ (for Theorem \ref{Th: M_smaller_N}): in the asymptotically high SNR regimes, transmission correlations always degrade the PEP performance. In this simulation, both $\mathbf{P}_1$ and $\mathbf{P}_2$ degrade the PEP performance.}\label{Fig: PEP_Correlation_Degradtion_M3N4}.
\end{figure}

It is  worth mentioning here that when the STC is orthogonal (i.e., all the eigenvalues of $\mathbf{\Delta}^H \mathbf{\Delta}$ are identical), Theorem \ref{Th: M_greater_N} will never be satisfied since Lemma \ref{Lemma: Bound} implies that the transmission correlations always degrade the PEP performance for orthogonal codes, which is consistent with the results in \cite{Yahampath2007} and \cite{Shin2008}, where majorization was used to show this property for orthogonal code in the MIMO single-keyhole channel.
Finally, we compare the effects of transmission correlations on PEP performances for the single-keyhole and Rayleigh channels in Table \ref{Tab: CorrelationAffect Comparision}. We can see that transmission correlations play different roles on the PEP performances in the two types of fading channels.


\begin{table*}[!t]
\centering
\caption{Effects of transmission correlations on the PEP performance for the single-keyhole and Rayleigh channels in the asymptotically high SNR regimes}
\begin{tabular}{|l|p{4cm}|p{6cm}|}
\hline
            & Rayleigh Channels                        & Single-keyhole Channels                    \\
\hline
OSTBCs      & Always degrades  the performance in high SNR regimes \cite{Wang2004} \cite{Bolcskei2000} & Always degrades the performance in high SNR regimes \cite{Yahampath2007} \cite{Shin2008}                                              \\
\hline
Non-orthogonal STCs               & Always degrades  the performance  in high SNR regimes \cite{Wang2004} \cite{Bolcskei2000} &  If $M > N$, may either degrade or improve in high SNR regimes (Our result, Theorem \ref{Th: M_greater_N} in this correspondence); if $M \le N$, always degrades in high SNR regimes (Our result, Theorem \ref{Th: M_smaller_N} in this correspondence).
                                                                    \\

\hline
\end{tabular}\label{Tab: CorrelationAffect Comparision}
\end{table*}

\section{Conclusion}\label{Sec: Conclusion}
In this paper, we investigated the effect of transmission correlations on the PEP performance of general STCs over single-keyhole channels. We proved that, in the asymptotically high SNR regimes, when $M \le N$, the transmission correlations always degrade the PEP performance; when $M>N$, depending on how the correlation matrix $\mathbf{P}$ beamforms the code word difference matrix $\mathbf{\Delta}$, the PEP performance of general STCs can be either degraded or improved. This property of the MIMO single-keyhole channel is more sophisticated than that of the MIMO Rayleigh channel. We also provided a new measure in \eqref{Eq: EigenValueCriteriaCorrelationAffect} to exam if there exists certain correlation matrices that can improve the performance of general STCs in high SNR regimes and we provided one form of such matrices if applicable. Our new measure on the general STCs can also be used to explain previous findings for OSTBCs.


%



\ifCLASSOPTIONcaptionsoff
  \newpage
\fi



%

\bibliographystyle{IEEEtran}
\bibliography{RFIDarticle2}

\end{document}